\documentclass[a4paper]{article}
\usepackage{fullpage}
\usepackage{amsthm,amsmath,amssymb}
\usepackage{bm}
\usepackage{algorithm, algpseudocode}
\algtext*{EndWhile}
\algtext*{EndIf}
\algtext*{EndFor}

\usepackage{graphicx}
\usepackage{enumitem}
\setlist{itemsep=0pt, topsep=3pt}

\renewcommand\vec[1]{\boldsymbol{#1}}

\graphicspath{{figs/}}

\newtheorem{theorem}{Theorem}[section]
\newtheorem{lemma}[theorem]{Lemma}
\newtheorem{claim}[theorem]{Claim}

\newcommand{\set}[1]{\left\{#1\right\}}

\DeclareMathOperator{\union}{\bigcup}

\newcommand{\R}{\mathbb{R}}
\newcommand{\Z}{\mathbb{Z}}

\newcommand{\calI}{{\mathcal I}}

\newcommand{\closed}{{\mathrm{closed}}}
\newcommand{\open}{{\mathrm{open}}}

\newcommand{\cost}{{\mathrm{cost}}}
\newcommand{\pcost}{{\mathrm{pcost}}}
\newcommand{\ccost}{{\mathrm{ccost}}}

\newcommand{\TC}{{\mathrm{TC}}} 

\title{On Facility Location with General Lower Bounds}
\author{Shi Li}
\date{}
\begin{document}

\maketitle

\begin{abstract}
	In this paper, we give the first constant approximation algorithm for the lower bounded facility location (LBFL) problem with general lower bounds. Prior to our work, such algorithms were only known for the special case where all facilities have the same lower bound: Svitkina \cite{Svi10} gave a $448$-approximation for the special case, and subsequently Ahmadian and Swamy \cite{AS13} improved the approximation factor to 82.6. 
	
	As in \cite{Svi10} and \cite{AS13}, our algorithm for LBFL with general lower bounds works by reducing the problem to the capacitated facility location (CFL) problem.  To handle some challenges caused by the general lower bounds, our algorithm involves more reduction steps.  
	One main complication is that after aggregation of clients and facilities at a few locations, each of these locations may contain many facilities with different opening costs and lower bounds.  To handle this issue, we introduce and reduce our LBFL problem to an intermediate problem called the transportation with configurable supplies and demands (TCSD) problem, which in turn can be reduced to the CFL problem.
\end{abstract}

\section{Introduction}
	We study the lower bounded facility location (LBFL) problem with general facility lower bounds. We are given a set $F$ of potential facility locations, a set $C$ of clients, a metric $d$ over $F \cup C$. Each facility $i \in F$ has an opening cost $f_i \geq 0$, and a lower bound $B_i \in \Z_{\geq 0}$ on the number of clients it must serve once it is opened.  The goal of the problem is to open some facilities and connect all clients to the open facilities, so as to minimize the sum of the opening cost and the connection cost.  Formally, a feasible solution to the problem is a pair $(S \subseteq F, \sigma \in S^C)$ such that for every $i \in S$, we have $|\{j \in C: \sigma_j = i\}|\geq B_i$. The goal is to minimize $\sum_{i \in S}f_i + \sum_{j \in C}d(j, \sigma_j)$.  

The problem was introduced independently by Guha et al.\ \cite{GMM00} and Karger and Minkoff \cite{KM00} as a subroutine to solve their buy-at-bulk network design problems.  The LBFL problem arises in this context since in near-optimal solutions, one needs to aggregate a certain amount of demands at a set of hub locations to avoid paying high fixed costs,  and at the same time make the cost of transporting demands small.  The uncapacitated facility location (UFL) problem, the special case of LBFL where all facilities $i$ have $B_i = 0$, is a classic problem in operations research and has been studied extensively in the literature. The lower bounds on facilities naturally arise in scenarios where a service can be provided only if there is enough demand. Then it is not surprising that the LBFL problem can find many direct applications.  

Since the special case UFL is already NP-hard, we aim to design efficient approximation algorithms for the LBFL problem. In their papers that introduced the problem, Guha et al.\ \cite{GMM00} and Karge and Minkoff \cite{KM00} developed an $O(1)$-bi-criteria approximation algorithm for LBFL that respect the lower bound constraints only approximately. Namely, the solution output by the algorithm has cost at most $O(1)$ times that of the optimum solution, and connects at least ${\beta} B_i$ clients to each open facility $i$, for some constant ${\beta} < 1$.  Such a bi-criteria approximation was sufficient for their purpose of solving the buy-at-bulk network design problems. True constant approximation algorithms are known for the special case of LBFL when all facilities have the same lower bound, i.e, $B_i = B$ for every $i \in F$. The first such algorithm is a 448-approximation algorithm due to Svitkina \cite{Svi08, Svi10}, which is based on reducing the LBFL problem to the capacitated facility location (CFL) problem. A remarkable feature of the reduction is that the roles of facilities and clients are reversed in the CFL instance.   The approximation ratio was later improved to 82.6 by Ahmadian and Swamy \cite{AS13}.  Both algorithms require the lower bounds to be uniform, and getting an $O(1)$-approximation for LBFL with general lower bounds remained an open problem, as discussed in both \cite{Svi10} and \cite{AS13}. 

In this paper, we solve the open problem in the affirmative:
\begin{theorem}\label{thm:main}
	There is a $4000$-approximation algorithm for the lower bounded facility location problem with general facility lower bounds.
\end{theorem}

\subsection{Related Work}
The related uncapacitated facility location (UFL) problem is one of the most classic problems studied in approximation algorithms and in operations research. There has been a long line of research on UFL \cite{STA97,JV99,CS03,KPR98,CG99,JMMSV03,JMS02,MYZ06,Byrka07,AGKMMP01} and almost all major techniques for approximation algorithms have been applied to the problem (see the book of Williamson and Shmoys \cite{WS11}). The current best approximation ratio for the problem is 1.488 due to Li \cite{Li11} and there is a hardness of 1.463 \cite{GK98}. 

The capacitated facility location (CFL) problem is the facility location problem where facilities have capacities (instead of lower bounds). That is, every facility $i$ has a capacity $u_i$ and if $i$ is open, then \emph{at most} $u_i$ clients can be connected to $i$.  The problem is motivated by the scenarios where a facility has limited resources and can only serve a certain number of clients when it is open.
P\'{a}l et al.\  \cite{PTW01} gave the first constant approximation algorithm for the problem,  with an approximation ratio of $9$.  The ratio has subsequently been improved in a sequence of papers \cite{MP03, ZCY05, BGG12}, with the current state-of-art ratio being the $5$ \cite{BGG12}.
The special case of CFL where all facilities have the same capacity has also been studied in the literature \cite{KPR98, CW99, AAB10}; it admits a better approximation ratio of $3$ \cite{AAB10}. All these algorithms for CFL are based on local search; the natural linear programming relaxation for the problem has unbounded integrality gap, and thus can only lead to $O(1)$-approximation for the soft-capacitated version\footnote{In this version, each facility can be opened multiple times but we pay the facility cost for each copy.} of the problem \cite{CS99, MYZ03}, and the special case where all facility costs are the same \cite{LSS12}.
In a recent breakthrough result, An et al.\ \cite{ASS14, ASS17} gave an LP-based $O(1)$-approximation for CFL, solving a long-standing open problem listed in the book of Williamson and Shmoys \cite{WS11}.


\subsection{Our Techniques} As in \cite{Svi10, AS13}, our algorithm reduces the LBFL problem to CFL, but it involves more reduction steps due to the general lower bounds.  As in \cite{Svi10, AS13}, we first run the bi-factor approximation algorithm in \cite{GMM00, KM00} to obtain an $O(1)$ approximate solution $(S^\circ, \sigma^\circ)$ where each open facility $i \in S^\circ$ is connected by at least ${\beta} B_i$ clients, for some ${\beta} \in (1/2, 1)$. We then obtain a more structured LBFL-instance $\calI^1$ by moving all clients to $S^\circ$ according to $\sigma^\circ$, and making facilities in $S^\circ$ free. The equivalence between $\calI$ and $\calI^1$ (up to an $O(1)$-loss in the approximation ratio) is straightforward and so we can focus on $\calI^1$ from now on. A crucial structure that $\calI^1$ has is that all clients are located at $S^\circ$, and for each $i \in S^\circ$, the number $n_i$ of clients at $i$ is at least ${\beta} B_i$.


For the uniform-lower-bound case, \cite{Svi10} showed the facilities not in $S^\circ$ can be removed, as opening a facility $i \notin S^\circ$ is not much better than opening the nearest neighbor of $i$ in $S^\circ$.  Then the residual problem becomes to decide which facilities in $S^\circ$ to open and how to connect clients. 
By viewing each client as a unit supply, \cite{Svi10} showed that $\calI^1$ can be converted to an instance of CFL. 
Roughly speaking, opening a facility in the instance $\calI^1$ corresponds to \emph{not opening} the correspondent supplier in the CFL instance. An open facility $i \in S^\circ$ in $\calI^1$ may need $t$ more connected clients to meet its lower bound; this corresponds to $t$ units of demand at $i$ in the CFL instance.

One complication for the general lower bound case is that facilities outside $S^\circ$ may be useful as they may have small lower bounds and opening them can avoid long connections.  We divide these facilities into two types and handle them separately. First, we show that facilities near $S^\circ$ can be moved to $S^\circ$, sacrificing only an $O(1)$-factor in the approximation ratio; the resulting instance will be an even more structured LBFL instance $\calI^2$.  Second, we construct an instance $\calI^3$ of what we call the LBFL with penalty problem. 
As a by-product of the formulation of $\calI^3$, facilities not collocated with $S^\circ$ in $\calI^2$ (i.e, facilities that are far away from $S^\circ$ in $\calI^1$) can be removed for free.

Here is how we construct the LBFL instance $\calI^2$. For each location $v \in S^\circ$, let $\ell_v = d(v, S^\circ \setminus\{v\})$ be the distance between $v$ and its nearest neighbor in $S^\circ$, and let $N_v = \set{i \in F: d(i, v) <\ell_v /2}$ be the set of facilities that are near $v$. Then $\calI^2$ is obtained by moving all facilities in $N_v$ to $v$, and changing the opening cost of $i \in N_v$ to $f_i + \Theta(n_v d(v, i))$. We show that an $O(1)$-approximate solution to $\calI^2$ leads to an $O(1)$-approximate solution to $\calI^1$.  Roughly speaking, moving a facility $i \in N_v$ to $v$ will not affect the cost of connecting $i$ to a client $j$ not at $v$ by too much. It decreases the distance between $i \in N_v$ and clients at $v$ to $0$; however, the decrease of distances can be charged using the $\Theta(n_vd(v, i))$ term in the opening cost of $i$ in $\calI^2$.

As mentioned, we then reduce the LBFL instance $\calI^2$ to an instance $\calI^3$ of the LBFL with penalty (LBFL-P) problem. 
$\calI^3$ has the same setting as $\calI^2$, but with the following differences.  In $\calI^3$, not all clients have to be connected. Instead, we impose a penalty of $\Theta(n_v\ell_v)$ for every $v \in S^\circ$ where no facility at $v$ (or equivalently, no facility in $N_v$) is open.  The penalty term makes the problem well-posed and non-trivial: in order to avoid high penalty, we may need to open some facilities, and to satisfy the lower bound requirements for these facilities, non-trivial connections may need to be made. As a by-product of the reduction, the facilities not collocated with $S^\circ$ in $\calI^2$ can be removed from $\calI^3$ since there is no need to open them.

A key to show the equivalence of $\calI^2$ and $\calI^3$ is a procedure that converts a solution to $\calI^3$ back to a solution to $\calI^2$; for the uniform-lower-bound case, such a procedure was given in \cite{Svi10}, though the LBFL-P instance $\calI^3$ was not explicitly defined in \cite{Svi10}.  Let $S^\circ_\open$ and $S^\circ_\closed$ be the set of locations in $S^\circ$ with and without open facilities respectively.  There might be some unconnected clients in $S^\circ_\closed$ in the solution for $\calI^3$. To connect these clients, we build a forest of trees over $S^\circ$, where we have an edge from each $v \in S^\circ_\closed$ to its nearest neighbor in $S^\circ$. We connect the unconnected clients by moving them upon the trees, and open a free facility $v \in S^\circ_{\closed}$ when we accumulated enough number of them. The incurred cost can be bounded by the sum of the penalty and connection cost of the solution to $\calI^3$.

Each location $v \in S^\circ$ in $\calI^3$ has many facilities, with different opening costs and lower bounds. We may open 1 facility at a location $v$; we may also choose to not open any facility at $v$, in which case we pay a penalty cost of $\Theta(n_v\ell_v)$. Then it immediately holds that $\calI^3$ is equivalent to an instance $\calI^4$ of what we call the transportation with configurable supplies and demands (TCSD) problem.  In the instance $\calI^4$, each location $v \in S^\circ$ has a set $R_v$ of choices, each being a pair $(g \in \Z_{\geq 0}, z \in \Z)$, which corresponds to putting $z$ units of net supply at location $v$ (if $z < 0$, putting $z$ units of net supply means putting $-z$ units of demand) at a cost of $g$.  Once we made the choices for all the locations in $S^\circ$, we solve the resulting transportation problem and pay the transportation cost.  Then the goal is to minimize the total cost we pay, including the cost for the choices and the transportation cost.   By setting the sets $R_v$'s naturally, one can see the equivalence between $\calI^3$ and $\calI^4$. This is the step where we switch the role of facilities and clients: a client in $\calI^3$ becomes a unit of supply in $\calI^4$.

With the TCSD instance $\calI^4$ defined, we can finally reach our CFL instance $\calI^5$.  By losing a factor of $2$, we assume all the costs in $\calI^4$ are integer powers of $2$; then for each $v \in S^\circ$ and a value $g$ which is power of $2$, we only need to keep the pair $(g, z)$ with the largest $z$.  This allows us to set up supplies and demands at each $v$ in the CFL instance $\calI^5$, so that the following happens. Losing another factor of $2$ in the approximation ratio, we can show that there is a one-to-one correspondence between the choices we have for $v$ in $\calI^4$ and those in $\calI^5$. So an $O(1)$-approximation for $\calI^5$ gives an $O(1)$-approximation for $\calI^4$, which leads all the way back to an $O(1)$-approximation for the original LBFL instance. 


\section{Notations and Useful Definitions}
For a metric $d$, a point $v$ and a set $V$ of points in the metric, we use $d(v, V)= \min_{u \in V}d(v, u)$ to denote the distance from $v$ to its nearest point in $V$. $F$ and $C$ are always the sets of facilities and clients in the original instance. For any vector $h \in \R^{F}$ and a subset $F' \subseteq F$ of facilities, we use $h(F') := \sum_{i \in F'}h_i$ to denote the sum of $h$ values over all facilities in $F'$. For a connection vector $\sigma \in (F \cup \{\bot\})^C$, and $i \in F \cup \{\bot\}$, we define $\sigma^{-1}(i):=\{j \in C:\sigma_j = i\}$ to be the set of clients assigned to $i$; here $\sigma_j = \bot$ indicates that $j$ is not connected in $\sigma$.

We shall use a tuple $(F, C, d, f, B)$ to denote an LBFL instance, where $F, C, d, f$ and $B$ are as in the description of the problem. 
Given an LBFL instance $\calI = (F, C, d, f, B)$, and a parameter ${\beta} \in [0 ,1]$, a ${\beta}$-covered solution to $\calI$ is a pair  $(S \subseteq F, \sigma \in S^C)$ such that for every $i \in S$, we have $|\sigma^{-1}(i)| \geq {\beta} B_i$. We simply say $(S, \sigma)$ is a (valid) solution to $\calI$ if it is a $1$-covered solution.  A (valid) solution to an UFL instance $\calI = (F, C, d, f)$ is a pair $(S \subseteq F, \sigma \in S^C)$.

Given an  LBFL instance $\calI = (F, C, d, f, B)$, and a connection vector $\sigma \in F^C$, we define $\ccost_{\calI}(\sigma) := \sum_{j \in C}d(j, \sigma_j)$ to be the connection cost of the vector $\sigma$. We use $\cost_{\calI}(S, \sigma):=f(S) + \ccost_{\calI}(\sigma)$  to denote the cost of a solution (or a ${\beta}$-covered solution ) $(S, \sigma)$ to $\calI$.  Given a UFL instance $\calI = (F, C, d, f)$, we define $\cost_{\calI}(S) = f(S) + \sum_{j \in C}d(j, S)$ to be the cost of the solution $S$ to $\calI$. Notice that for UFL, it suffices to use the set $S$ of open facilities to denote a solution. 

\section{The $O(1)$-Approximation Algorithm for LBFL}
%
%
In this section, we give our $O(1)$-approximation algorithm for LBFL. The algorithm works by performing a sequence of reductions that leads to the CFL problem eventually. Each reduction is  from one instance to the next in such a way that an $O(1)$-approximation for the latter implies an $O(1)$-approximation for the former.  In Section~\ref{subsec:bi-criteria}, we review the bi-criteria approximation algorithm of \cite{GMM00, KM00}, which we use to obtain a $\beta$-covered solution $(S^\circ, \sigma^\circ)$ with cost at most $O(1)$ times the cost of the optimum 1-covered solution, where $\beta$ is a parameter whose value will be set to $2/3$ in the end.  In Section~\ref{subsec:aggregating-clients}, we aggregate the clients by moving each client $j$ to the location $\sigma^\circ_j$; this gives our LBFL instance $\calI^1$. In Section~\ref{subsec:aggregating-facilities}, we aggregate nearby facilities of $S^\circ$ at $S^\circ$ to obtain our instance $\calI^2$. In Section~\ref{subsec:LBFL-P}, we construct our LBFL with penalty (LBFL-P) instance $\calI^3$, where we do not need to connect all clients, but pay penalty for ``not opening facilities''. In Section~\ref{subsec:TCSD} we reformulate the instance $\calI^3$ as an instance $\calI^4$ of the transportation with configurable supplies and demands (TCSD) problem. In Section~\ref{subsec:CFL} we reduce $\calI^4$ to the CFL instance $\calI^5$, for which $O(1)$-approximation algorithms are known. With all the reductions, we calculate the final approximation ratio for LBFL in Section~\ref{subsec:accounting}. For convenience, the factors lost in the reductions are given in Figure~\ref{fig:reductions}.

\begin{figure}[h]
	\centering
	\includegraphics[width=0.9\textwidth]{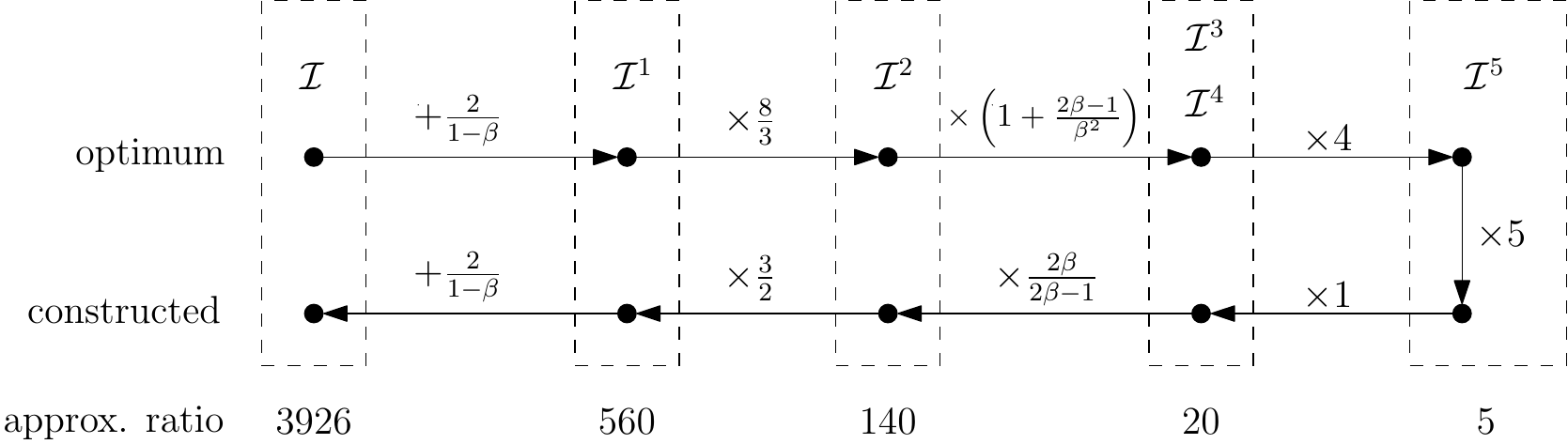}
	\caption{Factors we lose in the reductions. For each instance, the upper and lower dots stand for the optimum solution to the instance and the solution constructed by our algorithm respectively. Numbers after ``$\times$'' (``$+$'') indicates the multiplicative (additive) factor we lose when converting one solution to another. The approximation ratios (using $\beta = \frac23$) for all the instances are given in the bottom.} \label{fig:reductions}
\end{figure}

\subsection{Bi-Criteria-Approximation for LBFL via Reduction to UFL}
\label{subsec:bi-criteria}


	In this section, we apply the bi-criteria approximation algorithm of \cite{GMM00, KM00}, to obtain a ${\beta}$-covered solution $(S^\circ, {\sigma}^\circ)$ to the input LBFL instance ${\calI} = (F, C, d,  f,  B)$, where $\beta \in (1/2, 1)$ is a parameter whose value will be set to $2/3$ eventually. We give the algorithm for completeness.  Overall, we construct an auxiliary UFL instance $\calI' = (F, C, d, f')$ with some carefully designed opening cost vector $f'$, such that in any locally optimum solution $S$ to $\calI'$ under closing of facilities, every open facility $i$ is connected by at least $\beta B_i$ clients.
	
	The UFL instance ${\calI'} = (F, C, d,  f')$ has the same $F$, $C$ and $d$ as ${\calI}$, but facilities in $\calI'$ have different facility costs and no lower bounds. For every $i \in F$, let $J_i$ be the set of $B_i$ clients in $C$ nearest to $i$. For every $i \in F$, the facility cost of $i$ in instance $\calI'$ is defined as $f'_i :=  f_i + \frac{2{{\beta}}}{1-{{\beta}}}\sum_{j \in J_i}d(i, j)$. Lemma~\ref{lemma:I-to-I'} and \ref{lemma:I'-to-I} relate $\calI$ and $\calI'$ in both directions.

	\begin{lemma}
		\label{lemma:I-to-I'}
		Let $(S, \sigma)$ be any valid solution to $\calI$. Then $f'(S) \leq  f(S) + \frac{2\beta}{1-{{\beta}}}\ccost_{\calI}({\sigma})$.
	\end{lemma}
	\begin{proof}
		Every $i \in S$ is connected by at least $B_i$ clients in the solution $(S, \sigma)$ to ${\calI}$. Thus $\sum_{j \in J_i}d(i, j) \leq \sum_{j \in \sigma^{- 1}(i)} d(i, j)$. 
		\begin{flalign*}
			&& f'({S}) &= \sum_{i \in {S}}f'_i 
			= \sum_{i \in {S}}\left( f_i + \frac{2{{\beta}}}{1-{{\beta}}}\sum_{j \in J_i}d(i, j)\right) \leq  f(S) + \frac{2{\beta}}{1-{\beta}}\sum_{i \in S, j\in \sigma^{-1}(i)}d(i, j) && \\
			&& &=  f({S}) + \frac{2{{\beta}}}{1-{{\beta}}}\sum_{j\in C}d(j, \sigma_j) 
			=  f({S}) + \frac{2{{\beta}}}{1-{{\beta}}}\ccost_{\calI}(\sigma).&& \qedhere
		\end{flalign*}
	\end{proof}

\begin{lemma}
	\label{lemma:I'-to-I}
	Given any solution $S'$ to $\calI'$, we can efficiently find a ${{\beta}}$-covered solution 
	$(S, {\sigma})$ to ${\calI}$ such that $\sigma_j$ is the nearest facility to $j$ in $S$, and $\cost_{{\calI}}(S, {\sigma}) \leq \cost_{\calI'}(S')$.
\end{lemma}
\begin{proof}
	We start from the set $S = S'$. While there exists some $i \in S$ such that $\cost_{\calI'}(S \setminus \{i\}) \leq \cost_{\calI'}(S)$, we update $S \gets S \setminus \{i\}$. Thus, eventually, we obtain a locally optimal solution $S$ to $\calI'$ under closing of facilities.  Then our solution to ${\calI}$ is $(S, \sigma)$, where $\sigma$ is the vector connecting every $j \in C$ to its nearest facility in $S$. Clearly, $\cost_{{\calI}}(S, \sigma) \leq \cost_{\calI'}(S)$ since $f(S) \leq f'(S)$. During the local search step, we only decreased $\cost_{\calI'}(S)$; thus, $\cost_{{\calI}}(S, \sigma) \leq \cost_{\calI'}(S)\leq \cost_{\calI'}(S')$.
	
	It remains to show that in the solution $(S, {\sigma})$, every facility $i \in S$ is connected by at least ${{\beta}} B_i$ clients. Assume towards the contradiction that some $i \in S$ has $|\sigma^{-1}(i)| < {{\beta}} B_i$. Then there are at least $(1-{{\beta}})B_i$ clients in $J_i \setminus \sigma^{-1}(i)$. One client in $J_i \setminus \sigma^{-1}(i)$, say $j'$, has 
	\begin{align*}
		d(i, j') \leq \frac{1}{(1-{{\beta}})B_i}\sum_{j \in J_i}d(i, j).
	\end{align*}
	Since $j'$ is not connected to $i$ in the solution $(S, {\sigma})$, it must be connected to some other facility $i' \in S$ with $d(j', i') \leq d(j', i)$. Then, we consider the cost of connecting all clients in $\sigma^{-1}(i)$ to $i'$:
	\begin{align*}
		\sum_{j \in \sigma^{-1}(i)}d(j, i') &\leq \sum_{j \in \sigma^{-1}(i)} \left(d(j, i) + d(i, j') + d(j', i')\right) \leq \sum_{j \in \sigma^{-1}(i)} d(j, i) + |\sigma^{-1}(i)|\times 2d(i, j')\\
		&\leq \sum_{j \in \sigma^{-1}(i)} d(j, i) + {{\beta}} B_i\times \frac{2}{(1-{{\beta}})B_i}\sum_{j \in J_i}d(i, j)
		=\sum_{j \in \sigma^{-1}(i)} d(j, i) + \frac{2{{\beta}}}{1-{{\beta}}}\sum_{j \in J_i}d(i, j).
	\end{align*}
	
	Then we focus on the solution $S$ for the instance ${\calI'}$ and try to shut down the facility $i \in S$ and connect all the clients connected to $i$ to $i'$. (If twe connect these clients to their nearest facility in $S \setminus \{i\}$, the connection cost can only be smaller.)  The increase in the connection cost is at most $\frac{2{{\beta}}}{1-{{\beta}}}\sum_{j \in J_i}d(i, j)$, which is at most $f'_i$.  Thus, $\cost_{\calI'}(S \setminus\{i\}) \leq \cost_{\calI'}(S)$, contradicting the termination condition. Thus $(S, \sigma)$ is a ${\beta}$-covered solution to $\calI$. 
\end{proof}


Let $(S^*, {\sigma}^*)$ be the optimum solution to the LBFL instance $\calI$. Then, the above two lemmas lead to a bi-criteria approximation for LBFL:
\begin{lemma}
	\label{lemma:applying-UFL}
	We can efficiently find a  ${\beta}$-covered solution $(S^\circ, \sigma^\circ)$ to $\calI$ such that 
	\begin{align*}
		\cost_{\calI}(S^\circ, {\sigma}^\circ) \leq \frac{2}{1-\beta}\cost_{\calI}(S^*, \sigma^*).
	\end{align*}
	Moreover, $\sigma^\circ_j$ is the nearest facility in $S^\circ$ to $j$ for every $j \in C$.
\end{lemma}

\begin{proof}
	\cite{JMMSV03} gives an approximation algorithm for UFL, that outputs a solution $S'$ to $\calI'$ such that for every solution $S'^*$ to $\calI'$, we have $\cost_{\calI'}(S') \leq f'(S'^*) + 2\sum_{j \in C}d(j, S'^*)$. Running the algorithm to obtain a solution $S'$ and applying the inequality with $S'^*$ replaced by $S^*$, we have
	\begin{align*}
		\cost_{\calI'}(S') &\leq f'(S^*) + 2\sum_{j \in C}d(j, S^*)  \leq f(S^*) + \frac{2{\beta}}{1-{\beta}} \ccost_{\calI}({\sigma}^*) + 2\ccost_{\calI}(\sigma^*) \\
		&= f(S^*) + \frac{2}{1-\beta}\ccost_{\calI}(\sigma^*) \leq \frac{2}{1-\beta}\cost_{\calI}(S^*, \sigma^*).
	\end{align*}
	where the second inequality follows by applying Lemma~\ref{lemma:I-to-I'} with $(S, \sigma) = (S^*, \sigma^*)$.  Applying Lemma~\ref{lemma:I'-to-I}, we obtain a $\beta$-covered solution $(S^\circ, \sigma^\circ)$ to $\calI$ such that $\cost_{\calI}(S^\circ, \sigma^\circ) \leq \cost_{\calI'}(S')\leq \frac{2}{1-\beta}\cost_{\calI}(S^*, \sigma^*)$. By the lemma, $\sigma^\circ_j$ is the nearest facility in $S^\circ$ to $j$ for every $j \in C$. 
\end{proof}
So,  we can apply Lemma~\ref{lemma:applying-UFL} to obtain a ${\beta}$-covered solution $(S^\circ, \sigma^\circ)$ to $\calI$ satisfying the conditions stated in the theorem.  Without loss of generality, we assume any two different facilities $i, i' \in S^\circ$ have $d(i, i') > 0$.

\subsection{Aggregating Clients}
\label{subsec:aggregating-clients}


With the $\beta$-covered solution $(S^\circ, {\sigma}^\circ)$ defined, we perform the client aggregation step where we move all the clients to their respective nearest facilities in  $S^\circ$. We also make the facilities in $S^\circ$ free since we can afford to pay their opening costs.
Formally, our new LBFL instance is $\calI^1 = (F, C, d^1,  f^1,  B)$, where we have
\begin{itemize}
	\item $d^1(i, i') = d(i, i')$ for every $i, i' \in F$, $d^1(i, j) = d(i, \sigma^\circ_j)$ for every $i \in F, j \in C$ and $d^1(j, j') = d(\sigma^\circ_{j}, \sigma^\circ_{j'})$ for every $j, j' \in C$, 
	\item $f^1_i = 0$ if $i \in S^\circ$ and $f^1_i = f_i$ if $i \notin S^\circ$.
\end{itemize}

Notice that $\calI$ and  $\calI^1$ have the same $F, C$ and $ B$, so they have the same set of valid solutions. Since we moved each client $j$ by a distance of $d(j, \sigma^\circ_j)$, the following claim holds.
\begin{claim}
	\label{claim:I-and-I1}
	For every valid solution $(S, {\sigma})$ to $\calI$ and $\calI^1$, we have
	\begin{align*}
		\Big|\ccost_{\calI^1}({\sigma}) - \ccost_{\calI}({\sigma})\Big| \leq \sum_{j \in C}d(j, \sigma^\circ_j) = \ccost_{\calI}({\sigma}^\circ).
	\end{align*}
\end{claim}

\begin{theorem}
	\label{thm:I-and-I1}
	Given an $\alpha_1$-approximate solution to $\calI^1$,  we can efficiently find an $\alpha$-approximate solution to $\calI$, where $\alpha = \alpha_1\left(1 + \frac{2}{1-\beta}\right) + \frac{2}{1-\beta}$.
\end{theorem}
\begin{proof}
	By Lemma~\ref{lemma:applying-UFL} and Claim~\ref{claim:I-and-I1}, we have 
	\begin{align}
		\cost_{\calI^1}(S^*, {\sigma}^*) &= f^1(S^*) + \ccost_{\calI^1}({\sigma}^*) \leq f(S^*) + \ccost_{\calI}({\sigma}^*) + \ccost_{\calI}({\sigma}^\circ) \nonumber\\
		&=\cost_{\calI}(S^*, \sigma^*) + \ccost_{\calI}({\sigma}^\circ) \leq \left(1+\frac{2}{1-\beta}\right)\cost_{\calI}(S^*, \sigma^*).
		\label{inequ:cost-I1-S*}
	\end{align}
	Then, an $\alpha_1$-approximate solution $(S, {\sigma})$ to $\calI^1$ will have $\cost_{\calI^1}(S, {\sigma}) \leq \alpha_1\cdot\cost_{\calI^1}(S^*, {\sigma}^*)$. 
	\begin{flalign*}
		&& &\quad \cost_{\calI}(S, {\sigma}) &&\\
		&& &= f(S) + \ccost_{\calI}(S, {\sigma}) \quad \leq \quad f^1(S) + f(S^\circ) + \ccost_{\calI^1}({\sigma}) + \ccost_{\calI}({\sigma}^\circ) && \text{by Claim~\ref{claim:I-and-I1}}\\
		&& &= \cost_{\calI}(S^\circ, {\sigma}^\circ) + \cost_{\calI^1}(S, \sigma)
		\quad \leq \quad \cost_{\calI}(S^\circ, {\sigma}^\circ) + \alpha_1\cdot\cost_{\calI^1}(S^*, {\sigma}^*) &&  \\
		&& &\leq \frac{2}{1-\beta}\cost_{\calI}(S^*, \sigma^*) + \alpha_1 \left(1+\frac{2}{1-\beta}\right)\cost_{\calI}(S^*, \sigma^*)   && \text{by Lemma~\ref{lemma:applying-UFL} and \eqref{inequ:cost-I1-S*}}\\
		&& &= \left(\alpha_1 \left(1 + \frac{2}{1-\beta}\right) + \frac{2}{1-\beta}\right)\cost_{\calI}(S^*, \sigma^*) \ =\ \alpha\cdot\cost_{\calI}(S^*, {\sigma}^*).  && \qedhere
	\end{flalign*}
\end{proof}

Thus, from now on, we can focus on the instance $\calI^1 = (F, C, d^1, f^1, B)$, where all the clients are collocated with facilities in $S^\circ$ in $\calI^1$. Thereafter it is convenient for us to view $S^\circ$ as a set of locations, which will be denoted using $v, r$ and $r'$. Occasionally, we shall use the fact that each $v \in S^\circ$ is also a facility with 0 opening cost.

For every $v \in S^\circ$, clients in $\sigma^{\circ-\!1}(v)$ are at location  $v$ and let $n_v := |\sigma^{\circ-\!1}(v)|$ be the number of such clients. Thus $n_v \geq {\beta} B_v$ since $(S^\circ, \sigma^\circ)$ is a ${\beta}$-covered solution. 


\subsection{Aggregating Nearby Facilities}
\label{subsec:aggregating-facilities}

	In this section, we move facilities near $S^\circ$ to $S^\circ$. For every $v \in S^\circ$, we define $\ell_v = d(v, S^\circ \setminus v)>0$ to be the distance from $v$ to its nearest neighbor in $S^\circ$.\footnote{We assume $|S^\circ| \geq 2$; otherwise all the clients are at the same location and the instance $\calI^1$ is easy.} We define $N_v = \set{i \in F: d(v, i) < \ell_v/2}$ to be the set of facilities in the open ball of radius $\ell_v/2$ centered at $v$. It is then easy to see that the sets $\set{N_v: v \in S^\circ}$ are disjoint.  For every $i \in F$, let $\phi_{i}$ be the location $v \in S^\circ$ such that $i \in N_v$, or let $\phi_{i} = i$ if no such $v$ exists.  

We then construct a new LBFL instance $\calI^2$ by moving all the facilities in $N_v$ to $v$ and changing their facility costs.  Formally, our new LBFL instance is $\calI^2 = (F, C, d^2,  f^2,  B)$, where
\begin{itemize}
	\item  $d^2(i, i') = d^1(\phi_i, \phi_{i'})$ for every $i, i' \in F$, $d^2(i, j) = d^1(\phi_i, j)$ for every $i \in F, j \in C$ and $d^2(j, j') = d^1(j, j')$ for every $j, j' \in C$;
	\item $f^2_{i} = f^1_{i} + \frac23n_i d^1(v, i)$ for every $v \in S^\circ$ and $i \in N_v$; if $i \notin \union_{v \in S^\circ}N_v$, then we have $f^2_{i} = f^1_{i}$. Notice that every $v \in S^\circ$ is a facility with $f^2_v = f^1_v + n_vd^1(v, v) = 0$.
\end{itemize}

\begin{claim}
	\label{claim:d1-to-d2}
	For every $j \in C$, and every $i \in F$, we have $d^2(i, j) \leq 2d^1(i, j)$.
\end{claim}
\begin{proof}
	If $i \notin \union_{v \in S^\circ}N_v$, then $d^2(i, j) = d^1(\phi_i, j) = d^1(i, j)$.  Focus on some $v \in S^\circ$ and $i \in N_v$; thus $d^2(i, j) = d^1(v, j)$. Recall that every client $j$ is at some location in $S^\circ$ in the metric $d^1$. If a client $j \in C$ is at $v$, then $0 = d^1(v, j) \leq d^1(i, j)$; otherwise, we have $d^1(v, j) \leq d^1(v, i) + d^1 (i, j) \leq \frac{d^1(v, j)}{2} + d^1(i,j)$ by the definition of $\ell_v$ and $N_v$. This implies $d^1(v, j) \leq 2d^1(i, j)$.
	So, in any case we have $d^2(i, j) = d^1(v, j) \leq 2d^1(i, j)$.
\end{proof}

We can relate $\calI^1$ and $\calI^2$ in both directions.

\begin{lemma}
	\label{lemma:I1-to-I2}
	For every solution $(S, {\sigma})$ to $\calI^1$, there is a solution $(S', {\sigma}')$ to $\calI^2$ such that $\cost_{\calI^2}(S', \sigma') \leq \frac83\cost_{\calI^1}(S, \sigma)$.
\end{lemma}
\begin{proof}
	In the metric $d^2$, all facilities in $N_v$ are at $v$ for every $v \in S^\circ$.  Our solution $(S', {\sigma}')$ is obtained from $(S, {\sigma})$ by keeping only one open facility $i$ from $N_v$ with the smallest $d^1(v, i)$. Formally, let $(S',{\sigma}') = (S, {\sigma})$ initially and for every $v \in S^\circ$ we apply the following procedure.  If $|S' \cap N_v| \geq 2$ we then let $i$ be the facility in $S' \cap N_v$ with the smallest $d^1(v, i)$. We update $S' \gets S' \setminus N_v \cup \{i\}$, and for every $j \in C$ with $\sigma'_j \in N_v \setminus \{i\}$ we update $\sigma'_j$ to $i$.
	
	 The final solution $(S', {\sigma}')$ is valid to $\calI^2$; moreover, $\ccost_{\calI^2}({\sigma}') =\ccost_{\calI^2}({\sigma})$ since all facilities in $N_v$ are collocated at $v$ in the instance $\calI^2$.  By Claim~\ref{claim:d1-to-d2}, we have $\ccost_{\calI^2}({\sigma}') \leq 2\ccost_{\calI^1}(\sigma)$. We then consider the facility cost of solution $(S', {\sigma}')$.
	\begin{align*}
		f^2(S') - f^1(S') = \sum_{v \in S^\circ, i \in N_v \cap S'} (f^2_i - f^1_i) = \frac23\sum_{v \in S^\circ, i \in N_v \cap S'}n_v d^1(i, v) \leq \frac23\ccost_{\calI^1}({\sigma}).
	\end{align*}
	To see the inequality, recall that $|S' \cap N_v| \leq 1$ for every $v \in S^\circ$. Moreover,  if $i \in S' \cap N_v$, then $i$ is the nearest facility to $v$ in $S$ in the metric $d^1$. Thus, the $n_v$ clients at $v$ must have total connection cost at least $n_vd^1(i, v)$ in the solution $(S, {\sigma})$ to instance $\calI^1$. Thus,
	
	\begin{flalign*}
		&& \cost_{\calI^2}(S', {\sigma}') &= f^2(S') + \ccost_{\calI^2}(\sigma') \leq f^1(S') + \frac23\ccost_{\calI^1}({\sigma}) + 2\ccost_{\calI^1}({\sigma})&&\\
		&& &\leq f^1(S) + \frac83\ccost_{\calI^1}(\sigma) \leq \frac83\cost_{{{\calI^1}}}(S, {\sigma}). && \qedhere
	\end{flalign*}
\end{proof}

\begin{lemma}
	\label{lemma:I2-to-I1}
	Let $(S, {\sigma})$ be a valid solution to $\calI^2$. Then we have  $$		\cost_{\calI^1}(S, {\sigma}) \leq \frac{3}{2}\cost_{\calI^2}(S, {\sigma}).$$
\end{lemma}
\begin{proof}
	Focus on a client $j$ with $\sigma_j = i$. We consider 3 cases separately.
	\begin{itemize}
		\item If $i \notin \union_{v \in S^\circ}N_{v}$, then we have $d^1(i, j) = d^2(i, j)$. 
		\item $i \in N_{v}$ for some $v \in S^\circ$ and $j$ is not collocated with $v$ in the metric $d^1$. Then $d^1(i, j) \leq d^1(v, j) + d^1(i, v) \leq d^1(v, j) + \frac{d^1(v, j)}{2} = \frac{3d^1(v, j)}{2}$, by the fact that $j$ must be at a location in $S^\circ$ in the metric $d^1$, and the definition of $N_{v}$. Thus, $d^1(i, j) \leq \frac{3d^1(v, j)}{2}=\frac{3d^2(i, j)}2$ in this case.
 		\item $i \in N_{v}$ for some $v \in S^\circ$ and $j$ is collocated with $v$ in metric $d^1$. Then we have $d^1(i, j) = d^1(i, v)$ and $d^2(i, j) = 0$; but there are at most $n_{v}$ such clients. 
	\end{itemize}
	So, we have
	\begin{flalign*}
		&& \cost_{\calI^1}(S, {\sigma}) &= f^1(S) + \sum_{j \in C}d^1(j, \sigma_j) \leq f^1(S) + \frac{3}{2}\sum_{j \in C}d^2(j, \sigma_j) + \sum_{v \in S^\circ, i \in S \cap N_{v}}n_{v}d^1(i, v) &&\\
		&& & = \sum_{i \in S \setminus \union_{v \in S^\circ}N_v}f^1_i + \sum_{v \in S^\circ, i \in S \cap N_v}(f^1_i + n_v d^1(i,v)) + \frac{3}{2}\sum_{j \in C}d^2(j, \sigma_j) \\
		&& &= f^2\left(S \setminus \union_{v \in S^\circ}N_v\right) + \frac32f^2\left(S \cap \union_{v \in S^\circ}N_v\right) + \frac{3}{2}\ccost_{\calI^2}({\sigma}) &&\\
		&& &\leq \frac32f^2(S) + \frac{3}{2}\ccost_{\calI^2}({\sigma}) = \frac32\cost_{\calI^2}(S, \sigma). &&\qedhere
	\end{flalign*}
\end{proof}
	 
\begin{theorem}
	\label{thm:I1-and-I2}
	Given an $\alpha_2$-approximate solution to $\calI^2$, we can efficiently find an $(\alpha_1 = 4\alpha_2)$-approximate solution to $\calI^1$.
\end{theorem}
\begin{proof}
	Let $(S^{*1}, {\sigma}^{*1})$ be the optimum solution to the instance $\calI^1$.  By Lemma~\ref{lemma:I1-to-I2}, there exists a solution $(S, {\sigma})$ such that $\cost_{\calI^2}(S, {\sigma}) \leq \frac83\cost_{{{\calI^1}}}(S^{*1}, {\sigma}^{*1})$. Then any $\alpha_2$-approximate solution $(S', {\sigma}')$ to $\calI^2$ has $\cost_{\calI^2}(S', {\sigma}') \leq \frac{8\alpha_2}3\cost_{{{\calI^1}}}(S^{*1}, {\sigma}^{*1})$. Then by Lemma~\ref{lemma:I2-to-I1}, we have $\cost_{\calI^1}(S', {\sigma}') \leq \frac32\cost_{\calI^2}(S', \sigma') \leq 4\alpha_2\cost_{\calI^1}(S^{*1}, {\sigma}^{*1})$. Thus, we have found an $(\alpha_1 = 4\alpha_2)$-approximate solution to $\calI^1$.
\end{proof} 

Thus, it suffices for us to find an $\alpha_2$-approximate solution to the LBFL instance $\calI^2 = (F, C, d^2,  f^2,  B)$, whose properties are summarized here. In the instance,  there is a set $S^\circ \subseteq F$ of locations where all clients are located. For every $v \in S^\circ$, 
\begin{itemize}[topsep=3pt,itemsep=0pt]
	\item the opening cost of $v$ is $f^2_v = 0$,
	\item the set $\sigma^{\circ-\!1}(v)$ of clients are located at $v$, and $n_v = |\sigma^{\circ-\!1}(v)| \geq {\beta} B_v$,
	\item the set $N_v \subseteq F$ of facilities are located at $v$.
\end{itemize}
Moreover, for every facility $i \in F \setminus \union_{v \in S^\circ}N_v$, we have $d^2(i, v) \geq \ell_v/2$ for every $v \in S^\circ$; recall that $\ell_v = d^2(v, S^\circ \setminus \{v\}) = d^1(v, S^\circ \setminus\{v\})$ is the distance from $v$ to its nearest neighbor in $S^\circ$.

	\subsection{Constructing Instance $\calI^3$ of LBFL with Penalty}
	\label{subsec:LBFL-P}
	
	To convert the LBFL instance $\calI^2$ to a CFL instance, we construct an intermediate instance $\calI^3$, which has the same input as $\calI^2$. However, in $\calI^3$ we do not require all clients to be connected; instead we penalize locations in $S^\circ$ with no open facilities. 
	
	Formally, in the instance $\calI^3$, we have $F, C, d^2,  f^2,  B, S^\circ, \sigma^\circ, (n_v)_{v \in S^\circ}, (N_v)_{v \in S^\circ}, (\ell_v)_{v \in S^\circ}$ as defined in $\calI^2$, and they satisfy the same properties as they do in $\calI^2$.  The output of the problem is a pair $(S, {\sigma})$ where $S \subseteq F$ and $\sigma \in (S \cup \{\bot\})^C$ is the connection vector, and $\sigma_j = \bot$ indicates that $j$ is not connected. The lower bounds of open facilities need to be respected, namely, for every $i \in S$, we require $|\sigma^{-1}(i)| \geq B_i$.  The facility cost of the solution is $f^2(S)$ and the connection cost of the solution is $\ccost_{\calI^3}(S, {\sigma}) := \sum_{j \in C:\sigma_j \neq \bot}d^2(j, \sigma_j)$. In addition, we impose a penalty cost 
	$$\pcost_{\calI^3}(S) := \frac{2\beta-1}{2\beta^2}\sum_{v \in S^\circ:S \cap N_v = \emptyset}n_v \ell_v.$$ Namely, for any $v \in S^\circ$ such that no facility in $N_v$ is open, we need to pay a penalty of $\frac{2\beta-1}{2\beta^2}n_v\ell_v$. Then the overall cost of the solution $(S, \sigma)$ is $\cost_{\calI^3}(S, \sigma) := f^2(S) + \ccost_{\calI^3}(\sigma) + \pcost_{\calI^3}(S)$. The goal of the problem is to find a solution $(S, \sigma)$ with the minimum cost. We call $\calI^3$ a LBFL with penalty (LBFL-P) instance.

	Notice that in a solution $(S, \sigma)$ to $\calI^3$, there is no need to open a facility outside $\union_{v \in S^\circ}N_v$. If such a facility is open, we can simply shut it down and disconnect all its connected clients. This only decreases the facility and connection costs, and does not affect the penalty cost.  Also, we only need to open at most one facility in $N_v$ for any $v \in S^\circ$.
	
	One direction of the relationship between $\calI^2$ and $\calI^3$ is straightforward:
	\begin{claim}
		\label{claim:I2-to-I3}
		Let $\left(S, {\sigma}\right)$ be a solution to the LBFL instance ${\calI^2} = (F, C, d^2,  f^2,  B)$.  Then, $\left(S, {\sigma} \right)$ is also a valid solution to the LBFL-P instance $\calI^3$.  Moreover, we have $\cost_{{{\calI^3}}}\left(S, {\sigma} \right) \leq \left(1+\frac{2\beta-1}{\beta^2}\right)\cost_{{\calI^2}}\left(S, {\sigma}\right)$.
	\end{claim}
	\begin{proof}
		$(S, {\sigma})$ is clearly a valid solution to ${{\calI^3}}$. We bound the penalty cost of the solution to $\calI^3$: 
		\begin{align*}
			\pcost_{\calI^3}(S)&= \frac{2\beta-1}{2\beta^2}\sum_{v \in S^\circ: N_v \cap S = \emptyset}n_v \ell_v \leq \frac{2\beta-1}{2\beta^2}\sum_{v \in {S^\circ}: N_v \cap S = \emptyset}\sum_{j \in \sigma^{\circ-\!1}(v)} 2d^2(j, \sigma_j) \\
			&\leq \frac{2\beta-1}{\beta^2}\sum_{j \in C}d^2(j, \sigma_j) = \frac{2\beta-1}{\beta^2}\cdot\ccost_{\calI^2}(\sigma).
		\end{align*}
		To see the first inequality in the above sequence, notice that for every $v \in S^\circ$, all facilities outside $N_v$ have distance at least $\ell_v/2$ to $v$. Thus, if no facility in $N_v$ is open, then the connection cost of every client $j \in \sigma^{\circ-\!1}(v)$ is at least $\ell_v/2$.
		\begin{flalign*}
			\text{So}, && \cost_{\calI^3}(S, {\sigma}) &= f^2(S) + \ccost_{\calI^3}(\sigma) + \pcost_{\calI^3}(S) \leq f^2(S) + \ccost_{\calI^2}(\sigma) + \frac{2\beta-1}{\beta^2}\cdot \ccost_{\calI^2}(\sigma)&&\\
			&& &\leq \left(1+\frac{2\beta-1}{\beta^2}\right)\cost_{\calI^2}(S, {\sigma}). &&\qedhere
		\end{flalign*}
	\end{proof}
			
	The proof of the other direction of the relationship is more involved. Given a valid solution $(S, {\sigma})$ to $\calI^3$, we need to obtain a valid solution to $\calI^2$ of small cost, by connecting the unconnected clients in $\sigma$. We show that the incurred connection cost  can be bounded using the connection and penalty cost of the solution $(S, {\sigma})$ to $\calI^3$; this procedure is very similar to that in \cite{Svi10}.
		
	\begin{lemma}
		\label{lemma:I3-to-I2}
		Suppose we are given a solution $(S, {\sigma})$ to the LBFL-P instance ${\calI^3}$. Then we can efficiently construct a solution $(S', \sigma')$ to the LBFL instance ${\calI^2}$ such that $\ccost_{\calI^2}(S', \sigma') \leq \frac{2\beta}{2\beta-1}\cost_{{{\calI^3}}}(S, {\sigma})$.
	\end{lemma}
	
	\begin{proof}
		As discussed, we can assume $S \subseteq \union_{v \in S^\circ}N_v$ and $|S \cap N_v| \leq 1$ for every $v \in S^\circ$.  We need to show how to connect the clients in $\sigma^{-1}(\bot)$ at a small cost.  For every $v \in S^\circ$, let $n'_v = \big|\sigma^{\circ-\!1}(v) \cap \sigma^{-1}(\bot)\big|$ be the number of unconnected clients at $v$. Let ${S^\circ_\open} = \set{v \in {S^\circ}: S \cap N_v \neq \emptyset}$ be the set of locations in $S^\circ$ with open facilities and ${S^\circ_\closed} = {S^\circ} \setminus {S^\circ_\open}$ be the set of locations in $S^\circ$ without open facilities.  W.l.o.g, we assume $n'_v = 0$ for every $v \in S^\circ_\open$, since there is an open facility at every $v$. %
		
		During our process, we keep a number $\bar n_v$ of unconnected clients at $v$ for every $v\in S^\circ$; initially, $\bar n_v = n'_v$.  We then move these clients within $S^\circ$ by updating the $\bar n_v$ values accordingly. In the end, we guarantee that for every $v \in S^\circ_\closed$, we have either $\bar n_v \geq B_v$, in which case we open the free facility $v$, or $\bar n_v = 0$.  If $\bar n_v > 0$ for some $v \in S^\circ_\open$, we can simply connect the $\bar n_v$ clients to the open facility at $v$. Since the  newly open facilities are free, it suffices for us to bound the total moving distance of all clients.

		For every $v \in {S^\circ_\closed}$, we define $\pi_v = \arg\min_{v' \in {S^\circ} \setminus \{v\}}d^2(v, v')$ to be the nearest neighbour of $v$ in ${S^\circ}$; thus, $\ell_v = d^2(v, \pi_v)$. Define a directed graph $G$ as $({S^\circ}, \set{(v, \pi_v): v \in {S^\circ_\closed}})$. Then each weakly-connected component in $G$ is, 
		\begin{enumerate}[label=(\roman*)]
			\item either a directed tree with edges directed to the root,
			\item or (i) plus an edge from the root to some other vertex in the component.
		\end{enumerate}
		Case (ii) can be viewed as a directed tree with root being a cycle.  Using a consistent way to break ties, we can assume that the cycle contains exactly two edges.  For case (i), the root is in ${S^\circ_\open}$ and all other vertices in the component are in ${S^\circ_\closed}$; for case (ii), all vertices in the component are in ${S^\circ_\closed}$.   The two cases are depicted in Figure~\ref{fig:moving-tree}.
		\begin{figure}[h]
			\centering
			\includegraphics[width=0.65\textwidth]{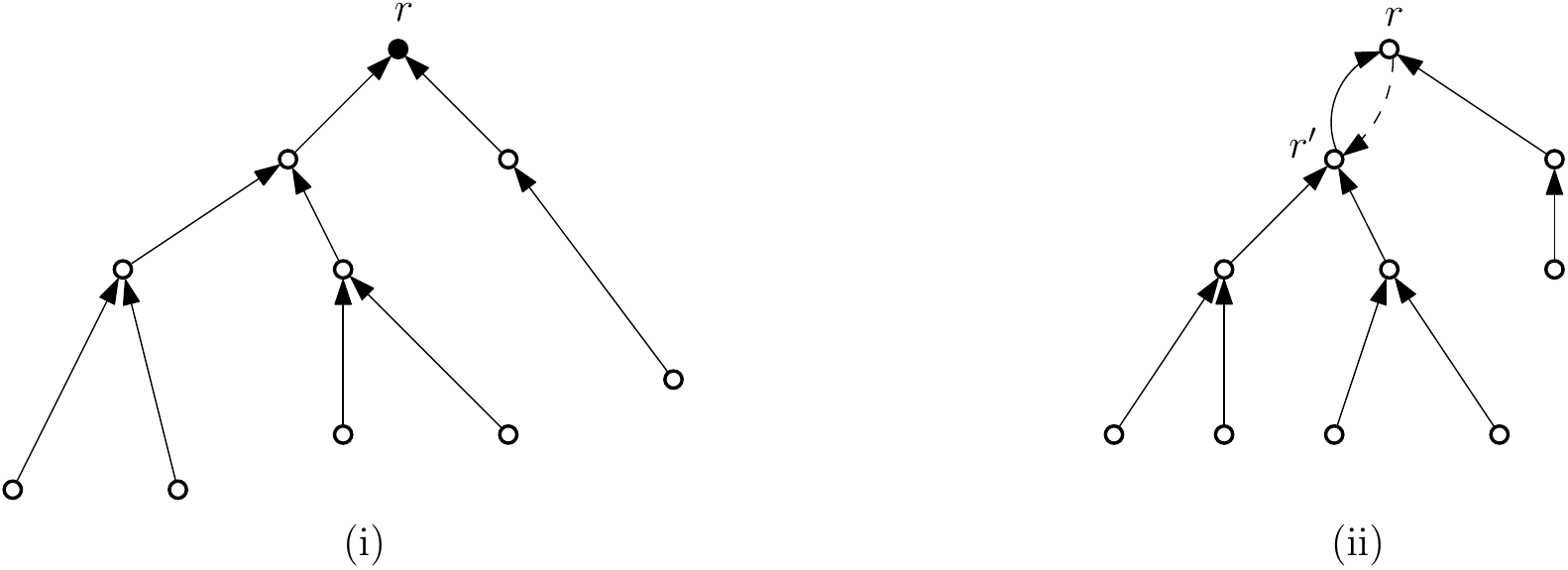}
			\caption{Two types of weakly-connected components in $G$. Every edge is from a vertex $v \in S^\circ_{\closed}$ to $\pi_v$. Empty and solid circles denote vertices in $S^\circ_\closed$ and $S^\circ_\open$ respectively.} \label{fig:moving-tree}
		\end{figure}

		To move the unconnected facilities, we handle each weakly-connected component of $G$ separately, in arbitrary order.  We now focus on a weakly-connected component of $G$. We first consider the simpler case (i), i.e, the component is a rooted tree. For every non-root vertex $v$ in the tree, from bottom to top, we perform the following operation. If $\bar n_v \geq B_{v}$, we open the free facility $v$. Otherwise, we move all the $\bar n_v$ unconnected clients at $v$ to the parent $\pi_v$ of $v$ in the tree. That is, we increase $\bar n_{\pi_v}$ by $\bar n_v$ and then change $\bar n_v$ to $0$.  The root $r$ of the component is in $S^\circ_\open$ and thus we can connect the $\bar n_r$ clients to the open facility at $r$.
		
		We then consider the more complicated case (ii). Suppose the length-2 cycle in the component is on two vertices $r, r' \in {S^\circ_\closed}$. By renaming we assume $B_{r} \leq B_{{r'}}$. By ignoring the edge from $r$ to $r'$, the component becomes a directed tree rooted at $r$. We then run the above procedure as for case (i). In the end, we shall handle the root $r$ as follows. If we have $\bar n_r \geq B_{r}$ then we open $r$; otherwise if ${r'}$ is open, we move the $\bar n_r$ clients at $r$ to $r'$. If neither condition holds, then we consider the vertex $v^* \in S^\circ_\open$ that is nearest to $\set{r, r'}$ and we move the $\bar n_r$ clients at $r$ to $v^*$. 
		
	 	It is easy to see after the moving process, we have for every $v \in {S^\circ_\closed}$, either $v$ is not open and $\bar n_v = 0$, or $v$ is open and $\bar n_v \geq B_{v}$. It remains to bound the moving cost incurred during the stage. Focus on a component in $G$ and again consider case (i) first. In this case, the number of clients moved along the edge $(v, \pi_v)$ is at most $B_{v} \leq \frac{n_v}{{{\beta}}}$, for every non-root $v$. Since $d^2(v, \pi_v) = \ell_v$, the total moving cost for these clients is at most $\frac{n_v\ell_v}{{{\beta}}}$.  
		
		For case (ii), we can also bound the cost of moving clients along an edge $(v, \pi_v)$ for any $v \neq r$ by $\frac{n_v\ell_v}{{{\beta}}}$. But additionally we need to consider the last step where we handle the root $r$. Consider the time point before we handle $r$. If we have $\bar n_r \geq B_{r}$ then no moving cost is incurred. Otherwise if ${r'}$ is open then we moved the clients from $r$ to $r'$ and the moving cost is $\bar n_r d(r, r') \leq B_{r}\ell_r \leq \frac{n_r\ell_r}{{{\beta}}}$.  
		
		It remains to consider the situation where  $\bar n_r < B_{r}$  and ${r'}$ is not open.  Notice that in this case the $n'_{r'} + n'_r$ unconnected clients  $(\sigma^{\circ-\!1}_r \cup \sigma^{\circ-\!1}_{r'}) \cap \sigma^{-1}(\bot)$ have been moved to $r$. So, we have that $B_{r} > \bar n_r \geq n'_r + n'_{r'}$.  Also $n_r + n_{r'} \geq {{\beta}} B_{r} + {{\beta}} B_{{r'}} \geq 2{{\beta}} B_{r}$, by our choice of $r$. So, there must be at least $n_r + n_{r'} - n'_r - n'_{r'} \geq  2{{\beta}} B_{r} - B_{r} = (2{{\beta}} - 1)B_{r}$ connected clients in $\sigma^{\circ-\!1}_r \cup \sigma^{\circ-\!1}_{r'}$ in $\sigma$. The connection cost of these clients in $\sigma$ is
		\begin{align*}
			\sum_{j \in \sigma^{\circ-\!1}_r \cup \sigma^{\circ-\!1}_{r'} \setminus \sigma^{-1}(\bot)} d^2(j, \sigma_j) \geq  (2{{\beta}}-1)B_{r}\cdot d^2(v^*, \set{r, r'}),
		\end{align*}
		since $v^*$ is the location containing the nearest facility to $\{r, r'\}$ in $S$ .
		The cost of moving the $\bar n_r < B_{r}$ clients from $r$ to $v^*$ is at most
		\begin{align}
			B_{r} d^2(r, v^*) &\leq B_{r} \left(d^2(v^*, \set{r, r'}) + \ell_r\right)
			\leq \frac{1}{2{{\beta}} - 1}\sum_{j \in \sigma^{\circ-\!1}_r \cup \sigma^{\circ-\!1}_{r'} \setminus \sigma^{-1}(\bot)} d^2(j, \sigma_j)  + \frac{n_r\ell_r}{{{\beta}}}. \label{inequ:moving-cost-for-r}
		\end{align}
		
		
		We can now bound the total moving cost of the unconnected clients. For every non-root $v$ in any component,  the cost for moving clients along edge $(v, \pi_v)$ is at most $\frac{n_v\ell_v}{{{\beta}}}$.  To handle the root $r$ for case (ii), the moving cost is most the right side of \eqref{inequ:moving-cost-for-r}. Thus, the total moving cost is at most
		\begin{align*}
			&\quad\frac{1}{{{\beta}}} \sum_{v \in {S^\circ_\closed}} n_v\ell_v + \frac{1}{2{{\beta}} - 1}\sum_{r, r'}\sum_{j \in \sigma^{\circ-\!1}_r \cup \sigma^{\circ-\!1}_{r'} \setminus \sigma^{-1}(\bot)} d^2(j, \sigma_j)\\
			&\leq \frac{1}{{{\beta}}}\sum_{v \in {S^\circ_\closed}}n_v\ell_v + \frac{1}{2{{\beta}} - 1}\sum_{j \in C \setminus \sigma^{-1}(\bot)} d^2(j, \sigma_j)
			= \frac{2\beta}{{2\beta-1}}\pcost_{\calI^3}(S) + \frac{1}{2{\beta} - 1}\ccost_{\calI^3}(\sigma).
		\end{align*}
		where $(r, r')$ is over all the length-2 cycles in components of case (ii).
		
		The process gives us the final solution $(S', \sigma')$ to $\calI^2$. We have that 
		\begin{flalign*}
			&& \cost_{\calI^2}(S', \sigma')&\leq f^2(S) + \ccost_{\calI^3}(\sigma) + \frac{2\beta}{2\beta-1}\pcost_{\calI^3}(S) + \frac{1}{2{{\beta}}-1}\ccost_{\calI^3}(\sigma) && \\
			&& &= f^2(S) + \frac{2\beta}{2\beta-1}\ccost_{\calI^3}(\sigma) + \frac{2\beta}{2\beta-1}\pcost_{\calI^3}(S) &&\\
			&& &\leq \frac{2{\beta}}{2{\beta}-1}\cost_{{{\calI^3}}}(S, {\sigma}). &&\qedhere
		\end{flalign*} 
	\end{proof}
	
	With the two lemmas, we can reduce the instance $\calI^2$ to $\calI^3$. 
	\begin{theorem}
		\label{thm:I2-and-I3}
		Given an $\alpha_3$-approximate solution to $\calI^3$, we can efficiently find an $\alpha_2$-approximate solution to $\calI^2$, where $\alpha_2 =\left(\frac{2\beta}{2\beta-1} + \frac{2}{\beta}\right)\alpha_3$.
	\end{theorem}
	\begin{proof}
		Let $(S^{*2}, \sigma^{*2})$ be the optimum solution to $\calI^2$. By Claim~\ref{claim:I2-to-I3}, we have that $\cost_{\calI^3}(S^{*2}, \sigma^{*2}) \leq \left(1+\frac{2\beta-1}{\beta^2}\right)\cost_{\calI^2}(S^{*2}, {\sigma}^{*2})$. So an $\alpha_3$-approximate solution $(S, {\sigma})$ to $\calI^3$ has $\cost_{\calI^3}(S, {\sigma}) \leq \left(1+\frac{2\beta-1}{\beta^2}\right)\alpha_3\cdot\cost_{{{\calI^2}}}(S^{*2}, {\sigma}^{*2})$. Then by Lemma~\ref{lemma:I3-to-I2}, we can find a solution $(S', \sigma')$ to $\calI^2$ such that $\cost_{\calI^2}(S', {\sigma}') \leq \frac{2{\beta}}{2{\beta}-1}\cost_{\calI^3}(S, \sigma) \leq \left(\frac{2\beta}{2\beta-1} + \frac{2}{\beta}\right)\alpha_3\cost_{\calI^2}(S^{*2}, {\sigma}^{*2})$, which is an $\left(\alpha_2 = \left(\frac{2\beta}{2\beta-1} + \frac{2}{\beta}\right)\alpha_3\right)$-approximate solution to $\calI^2$.
	\end{proof}
	
	Thus, it suffices to focus on the instance $\calI^3$ from now on. As discussed, we can remove facilities outside $\union_{v \in S^\circ}N_v$.

	\subsection{Transportation with Configurable Supplies and Demands Problem}
	\label{subsec:TCSD}
	We show that the LBFL-P instance $\calI^3$ is equivalent to an instance $\calI^4$ of the transportation with configurable supplies and demands (TCSD) problem. We describe the instance $\calI^4$ directly, since it is the only TCSD instance we deal with. We are given the metric $(S^\circ, d^2)$, $(n_v)_{v \in S^\circ}$ and $(N_v)_{v \in S^\circ}$ as in $\calI^3$. 
	For each $v \in {S^\circ}$, we define a set $R_v$ of pairs in $\Z_{\geq 0} \times \Z$ as follows: 
	\begin{align*}
		R_v = \set{\left(\frac{2\beta-1}{2\beta^2}n_v\ell_v, n_v\right)} \cup \set{(f^2_i, n_v - B_i):i \in N_v}.
	\end{align*}
	
	In the output to the instance $\calI^4$, we need to specify a pair $(g_v, z_v) \in R_v$ for every $v \in S^\circ$. If at a location $v$, we have $z_v \geq 0$, then we have $z_v$ units of supply at $v$; if $z_v < 0$, then we have $z_v$ units of demand at $v$. To be able to satisfy all the demands, we require $\sum_{v \in S^\circ}z_v \geq 0$. Then, the goal of the problem is to minimize $\cost_{\calI^4}(g, z):= g(S^\circ) + \TC_{d_2}(z)$, where $\TC_{d_2}(z)$ is the minimum  transportation cost for satisfying all the demands using the supplies.  Formally, 
		\begin{align*}
			\TC_{d_2}(z) := \min_{\psi} \sum_{v, v' \in {S^\circ}}\psi(v, v')d^2(v, v'),
		\end{align*} 
		where $\psi$ under the ``min'' operator is over all functions $\psi: {S^\circ} \times {S^\circ} \to \Z_{\geq 0}$ satisfying 
		\begin{align*}
			z_v + \sum_{v' \in {S^\circ}}\psi(v', v) - \sum_{v' \in {S^\circ}}\psi(v, v') \geq 0, \forall v \in {S^\circ}.
		\end{align*}
	We denote the instance $\calI^4$ as $(S^\circ, d^2,  \vec R = (R_v)_{v \in S^\circ})$.  We call the transportation with configurable supplies and demands problem since we need to choose a configuration of supplies and demands and then solve a transportation problem for the configuration.
	
	Recall that in a solution $(S, \sigma)$ to $\calI^3$, we open a set of facilities $S \subseteq \union_{v \in S^\circ}N_v$ satisfying $\forall v \in S^\circ, |S \cap N_v|\leq 1$, and connect some clients to $S$ so as to satisfy the lower bound constraints for facilities in $S$; we do not need to connect all clients. For every $v \in S^\circ$, if no facility in $N_v$ is open, a penalty cost of $\frac{2\beta-1}{2\beta^2}n_v\ell_v$ is incurred.   The goal is to minimize the sum of facility cost, connection cost and penalty cost.
	
	The equivalence between $\calI^3$ and $\calI^4$ can be easily seen by treating each client in the instance $\calI^3$ as a unit of supply in $\calI^4$.  For each $v \in S^\circ$, we may open $0$ or $1$ facility in $N_v$ in the solution $(S, \sigma)$ for $\calI^3$. If we do not open any facility, then we need to pay a penalty cost of $\frac{2\beta-1}{2\beta^2}n_v\ell_v$ and the $n_v$ clients at $v$ can provide $n_v$ units of supply; thus we have $\left(\frac{2\beta-1}{2\beta^2}n_v\ell_v, n_v\right) \in R_v$.  If we open $i \in N_v$, then we need to pay a facility cost of $f^2_i$. Depending on whether $n_v \geq B_i$ or not, we may either have $n_v - B_i$ units of supply, or $B_i - n_v$ units of demand at $v$. In either case, the correspondent pair is $(f^2_i, n_v - B_i)$. All the demands in $\calI^4$ have to be satisfied, since we must have $B_i$ clients connected to an open facility $i$ in a solution to $\calI^3$; but some units of supply may be unused, since we do not need to connect all the clients in $\calI^3$.

	Thus, any $\alpha_3$-approximate solution to $\calI^3$ corresponds to an $(\alpha_4 = \alpha_3)$-approximate solution to $\calI^4$. From now on, we focus on the TCSD instance $\calI^4=(S^\circ, d^2, \vec R = (R_v)_{v \in S^\circ})$.  Notice that for every $v \in S^\circ$, we have that $(f^2_v = 0, n_v - B_v) \in R_v$. 
		
\subsection{Converting TCSD instance $\calI^4$ to a CFL instance $\calI^5$}
\label{subsec:CFL}
	In this section, we shall show how to convert the TCSD instance $\calI^4 = ({S^\circ}, d^2,  \vec R)$ to a CFL instance $\calI^5$. To avoid confusion, we shall use demands and suppliers to denote clients and facilities, when we describe the CFL instance. We first round up each non-zero $g$ value in $\vec R$ to the nearest integer power of $2$. From now on we assume that for every $v \in {S^\circ}$ and every pair $(g, z) \in R_v$, we have either $g = 0$ or $g$ is an integer power of $2$. This incurs a factor of $2$ loss that we take into account in Theorem~\ref{thm:I4-and-I5}. If there are two different pairs $(g, z), (g', z') \in R_v$ with $g \leq g'$ and $z \geq z'$, then we can simply remove $(g', z')$ from $R_v$. Thus, we can assume that we can order the pairs in $R_v$ such that both the $g$ and $z$ values in the ordering are strictly increasing.

	Now we are going to construct our CFL instance $\calI^5$. The metric space for the CFL instance is $({S^\circ}, d^2)$ and we need to specify the demands and suppliers to put on each location $v \in {S^\circ}$.  Focus on a location $v \in {S^\circ}$ and the set $R_v$ of pairs. Assume $R_v = \{(h^v_\ell, y^v_\ell): \ell \in [L^v]\}$ for some $L^v \geq 1$. We assume $ 0 = h^v_1 < h^v_2 < h^v_3 < \cdots < h^v_{L^v}$ (recall that $(0, z) \in R_v$ for some $z$) and $y^v_1 < y^v_2 < y^v_3 < \cdots < y^v_{L^v}$ ($y$ values may be negative).   If $y^v_1 \leq 0$, then we put $|y^v_1| =  - y^v_1$ units of demand at location $v$; otherwise,  we build at $v$ a free supplier $k^v_1$ with $y^v_1$ units of supply.  For every $\ell =2, 3, \cdots, L^v$, we build at $v$ a supplier $k^v_\ell$ of cost $h^v_\ell$ with $y^v_\ell - y^v_{\ell-1}$ units of supply.  This finishes the construction of the CFL instance.
		
	\begin{lemma}
		\label{lemma:I4-to-I5}
		For every solution $( g,  z)$ to $\calI^4$, there is a solution to $\calI^5$ whose cost is at most $2\cost_\calI^2( g,  z)$,
	\end{lemma}
	
	\begin{proof}
		Focus on some $v \in {S^\circ}$ and assume $(g_v, z_v) = (h^v_\ell, y^v_\ell)$. In the solution for $\calI^5$,
		we open the free supplier $k^v_1$ if it exists, and we open $k^v_{\ell'}$ for every $\ell' = 2, 3, \cdots, \ell$. It is easy to see that at $v$, the net-supply is exactly $y^v_1 + (y^v_2 - y^v_1) + (y^v_3 - y^v_2) + \cdots + (y^v_{\ell} - y^v_{\ell-1}) = y^v_{\ell} = z_v$.   Thus, the connection cost of the solution for $\calI^5$ is exactly $\TC_{d_2}(z)$.  Also, the cost for opening all suppliers at $v$ is at most $\sum_{\ell' = 2}^\ell h^v_{\ell'} \leq 2g^v_\ell = 2g_v$; we used the fact that $h^v_2 < h^v_3 < \cdots < h^v_\ell$ and the numbers are all integer powers of $2$.  
		This finishes the proof of the lemma.
	\end{proof}
	
	\begin{lemma}
		\label{lemma:I5-to-I4}
		Any solution to $\calI^5$ can be converted to a solution $( g,  z)$ to $\calI^4$ with cost at most that of the solution to $\calI^5$.  
	\end{lemma}
	\begin{proof}
		Now we assume we are given  a solution to the CFL instance $\calI^5$.  Focus on a location $v \in {S^\circ}$. Let $\ell = \max\set{\ell: k^v_\ell\text{ is open}}$ or define $\ell = 1$ if no suppliers at $v$ is open in the solution for $\calI^5$. Then we choose $(g_v, z_v) = (h^v_\ell, y^v_\ell)$ in the solution $( g,  z)$ for $\calI^4$.  Notice that the net supply at $v$ in the solution $( g,  z)$ for the TCSD instance $\calI^4$ is $y^v_\ell = y^v_1 + (y^v_2 - y^v_1) + (y^v_3 - y^v_2) + \cdots + (y^v_{\ell} - y^v_{\ell-1})$, which is at least the net supply at $v$ in the solution for the CFL instance $\calI^5$, by our definition of $\ell$. Also, $g_v = h^v_\ell$, which is at most total cost of the opening suppliers at location $v$ in the solution for $\calI^5$.  Thus, we have that $\TC_{d_2}(z)$ is at most the connection cost of the solution for $\calI^5$, and $g(S^\circ)$ is at most the supplier opening cost. This finishes the proof of the Lemma.	
	\end{proof}
	
	\begin{theorem}
		\label{thm:I4-and-I5}
		We can efficiently find an $(\alpha_4 = 4\alpha_{\mathrm{CFL}})$-approximate solution to $\calI^4$.
	\end{theorem}
	\begin{proof}
		Let $(g^*, z^*)$ be the optimum solution to the instance $\calI^4$. By Lemma~\ref{lemma:I4-to-I5},  there is a solution to $\calI^5$ with cost at most $2\cost_{\calI^4}(g^*, z^*)$. Then an $\alpha_{\mathrm{CFL}}$-approximate solution to $\calI^5$ has cost at most $2\alpha_{\mathrm{CFL}}\cost_{\calI^4}(g^*, z^*)$. By Lemma~\ref{lemma:I5-to-I4}, we can efficiently find a solution $(g, z)$ to $\calI^4$ with $\cost_{\calI^4}(g, z) \leq 2\alpha_{\mathrm{CFL}}\cost_{\calI^4}(g^*, z^*)$.  This is a $2\alpha_{\mathrm{CFL}}$-approximate solution. 
		
		Considering the factor of 2 incurred by rounding the costs in $\calI^4$ to the integer powers of $2$, the finaly approximation ratio we obtain for $\calI^4$ is $\alpha_4 = 4\alpha_{\mathrm{CFL}}$.
	\end{proof}

	\subsection{Combining Everything} \label{subsec:accounting}
	We use the algorithm of \cite{PTW01} to solve the CFL instance $\calI^5$ to obtain an $(\alpha_{\textrm{CFL}} = 5)$-approximation for the instance. We set ${\beta} = \frac23$. Applying Theorems~\ref{thm:I4-and-I5}, \ref{thm:I2-and-I3}, \ref{thm:I1-and-I2}, and \ref{thm:I-and-I1}, and the fact that $\calI^3$ and $\calI^4$ are equivalent, the approximation ratios for all instances in our reduction are as follows:
	\begin{alignat*}{2}
		\alpha_3 &= \alpha_4 = 4\alpha_{\mathrm{CFL}} = 4 \times 5 = 20, &\qquad \alpha_2 &= \left(\frac{2\beta}{2\beta-1} + \frac{2}{\beta}\right)\alpha_3= 7 \times 20 = 140,\\
		\alpha_1 &= 4\alpha_2 = 560, &\qquad \alpha&=\alpha_1\left(1 + \frac{2}{1-\beta}\right) + \frac{2}{1-\beta} = 560 \times 7 + 6 \leq 4000.
	\end{alignat*}
	This finishes the proof of Theorem~\ref{thm:main}.

\section{Conclusion}
	In this paper, we developed a $4000$-approximation algorithm for the lower bounded facility location (LBFL) problem with general lower bounds. The algorithm  reduces the LBFL problem to the capacitated facility location (CFL) problem via a sequence of reductions. When describing the algorithm, we focused more on cleanness of presentation, rather than optimizing the final approximation ratio.  So we make all the reductions in a transparent way. It is possible to obtain better approximation ratio by considering the structures of the intermediate instances and analyzing the factors lost jointly. However this will inevitably complicate the algorithm and analysis; even with the complications, it seems hard to use this approach to improve the approximation ratio for LBFL to below 100. 
	
	To obtain a small constant approximation ratio for the LBFL problem, one interesting direction to pursue is to design a simple LP-based algorithm, without going through so many reductions. The natural LP relaxation for the problem has unbounded integrality gap, as shown in \cite{AS13}; so stronger LP relaxations are needed for this task. Using our reduction from LBFL to CFL, and the LP-based $O(1)$-approximation for CFL due to An et al.\ \cite{ASS17}, one could obtain an LP-based algorithm for LBFL in a mechanical way. Such an algorithm could serve as a useful baseline for us to understand the challenges of designing LP-based algorithms for LBFL. 
	
	\bibliographystyle{plain}
	\bibliography{reflist}
	
\end{document}